\documentclass[letterpaper]{sig-alternate-pages}
\usepackage{hyperref,breakurl}
\usepackage{stmaryrd}
\usepackage{aodv}
\usepackage{url}
\usepackage{amsfonts,amsmath,amssymb}
\usepackage{algorithm,algorithmic}
\floatname{algorithm}{Process}

\usepackage{graphicx}
\usepackage{color}
\usepackage{pdfsync}
\usepackage[caption=false]{subfig}
\usepackage{xspace}
\usepackage{enumerate}

\newcommand{\conditional}{conditional\xspace}

\newtheorem{common}{Common}[section]{\bfseries}{\itshape}%
\newcommand{\numberedthing}[2]{%
     \newtheorem{#1}[common]{#2}{\bfseries}{\itshape}}%
\numberedthing{theorem}{Theorem}%
\numberedthing{lemma}{Lemma}%
\numberedthing{prop}{Proposition}%
\numberedthing{remark}{Remark}%
\newtheorem{exam}[common]{Example}%

\newtheorem{definition}[common]{Definition}

\def\squareforqed{\hbox{\rlap{$\sqcap$}$\sqcup$}}
\def\endbox{\ifmmode\squareforqed\else{\unskip\nobreak\hfil
\penalty50\hskip1em\null\nobreak\hfil\squareforqed
\parfillskip=0pt\finalhyphendemerits=0\endgraf}\fi}

\newcommand{\plat}[1]{\raisebox{0pt}[0pt][0pt]{#1}} 
\newcommand{\spaces}[1]{\ #1\ }
\newcommand{\ans}{\spaces{\wedge}}

\newcommand{\awn}{AWN\xspace}
\newcommand{\datapackets}{data packets}
\newcommand{\dsnf}{sequence-number-status flag\xspace}
\newcommand{\rte}{routing table entry\xspace}
\newcommand{\rtes}{routing table entries\xspace}

\newenvironment{simpleProcess}{%
  \renewcommand{\algindent}{0.1em}
  \renewcommand{\algorithmicelsif}{$+$\algorithmicif}
  \renewcommand{\algorithmiccomment}[1]{\textcolor{blue}{\hspace{\algindent}/*\,##1\,*/}}
  \algsetup{indent=0.7em}
  \begin{algorithmic}%
  }{
  \end{algorithmic}
  }

\makeatletter
\def\comesfrom{\@transition\leftarrowfill}
\def\goesto{\@transition\rightarrowfill}
\def\ngoesto{\@transition\nrightarrowfill}
\def\Goesto{\@transition\Rightarrowfill}
\def\nGoesto{\@transition\nRightarrowfill}
\def\xmapsto{\@transition\mapstofill}
\def\nxmapsto{\@transition\nmapstofill}
\def\@transition#1{\@@transition{#1}}
\newbox\@transbox
\newbox\@arrowbox
\newbox\@downbox
\def\@@transition#1#2%
   {\setbox\@transbox\hbox
      {\vrule height 1.5ex depth .9ex width 0ex\hskip0.25em$\scriptstyle#2$\hskip0.25em}
   \ifdim\wd\@transbox<1.5em
      \setbox\@transbox\hbox to 1.5em{\hfil\box\@transbox\hfil}\fi
   \setbox\@arrowbox\hbox to \wd\@transbox{#1}
   \ht\@arrowbox\z@\dp\@arrowbox\z@
   \setbox\@transbox\hbox{$\mathop{\box\@arrowbox}\limits^{\box\@transbox}$}
   \dp\@transbox\z@\ht\@transbox 10pt
   \mathrel{\box\@transbox}}
\def\nrightarrowfill{$\m@th\mathord-\mkern-6mu%
  \cleaders\hbox{$\mkern-2mu\mathord-\mkern-2mu$}\hfill
  \mkern-6mu\mathord\not\mkern-2mu\mathord\rightarrow$}
\def\Rightarrowfill{$\m@th\mathord=\mkern-6mu%
  \cleaders\hbox{$\mkern-2mu\mathord=\mkern-2mu$}\hfill
  \mkern-6mu\mathord\Rightarrow$}
\def\nRightarrowfill{$\m@th\mathord=\mkern-6mu%
  \cleaders\hbox{$\mkern-2mu\mathord=\mkern-2mu$}\hfill
  \mkern-6mu\mathord\not\mathord\Rightarrow$}
\def\mapstofill{$\m@th\mathord\mapstochar\mathord-\mkern-6mu%
  \cleaders\hbox{$\mkern-2mu\mathord-\mkern-2mu$}\hfill
  \mkern-6mu\mathord\rightarrow$}
\def\nmapstofill{$\m@th\mathord\mapstochar\mathord-\mkern-6mu%
  \cleaders\hbox{$\mkern-2mu\mathord-\mkern-2mu$}\hfill
  \mkern-6mu\mathord\not\mkern-2mu\mathord\rightarrow$}
\makeatother 
\newcommand{\ar}[1]{\mathrel{\goesto{#1}}}            

\begin{document}%
\global\boilerplate={}
\global\copyrightetc{doi:\href{http://dx.doi.org/10.1145/2387238.2387274}{10.1145/2387238.2387274}}
\conferenceinfo{MSWiM'12,} {October 21--25, 2012, Paphos, Cyprus.} 
\CopyrightYear{2012} 
\crdata{978-1-4503-1628-6/12/10} 
\clubpenalty=10000 
\widowpenalty = 10000
\title{A Rigorous Analysis of AODV and its Variants}%
\numberofauthors{6}
\author{
\alignauthor Peter H{\"o}fner\\
       \affaddr{NICTA, Australia}\\[0.5mm]
       \affaddr{University~of~New~South~Wales,}\\
       \affaddr{Australia}\\
       \affaddr{\tt Peter.Hoefner@nicta.com.au}
\alignauthor Wee Lum Tan\\
       \affaddr{NICTA, Australia}\\[0.5mm]
       \affaddr{University~of~Queensland,}\\
       \affaddr{Australia}\\
       \affaddr{\tt WeeLum.Tan@nicta.com.au}
\alignauthor Annabelle McIver\
       \affaddr{Macquarie University, Australia}\\[0.5mm]
       \affaddr{NICTA, Australia}\\
       \affaddr{\tt annabelle.mciver@mq.edu.au}
\and  
\alignauthor Rob van Glabbeek\\
       \affaddr{NICTA, Australia}\\[0.5mm]
       \affaddr{University~of~New~South~Wales,}\\
       \affaddr{Australia}\\
       \affaddr{\tt rvg@cs.stanford.edu}
\alignauthor Marius Portmann\\
       \affaddr{NICTA, Australia}\\[0.5mm]
       \affaddr{University~of~Queensland,}\\
       \affaddr{Australia}\\
       \affaddr{\tt marius@itee.uq.edu.au}
\alignauthor Ansgar Fehnker\\
       \affaddr{NICTA, Australia}\\[0.5mm]
       \affaddr{University~of~New~South~Wales,}\\
       \affaddr{Australia}\\
       \affaddr{\tt Ansgar.Fehnker@nicta.com.au}
}

\maketitle
\vspace{-24pt} 
\begin{abstract}
In this paper we present a rigorous analysis of the Ad hoc On-Demand
Distance Vector (AODV) routing protocol using a formal specification  in {\awn} (Algebra for
Wireless Networks), a process algebra which has been specifically
tailored for the modelling of Mobile Ad Hoc Networks and Wireless Mesh Network
protocols.  Our formalisation models the exact details of the core
functionality of AODV, such as route discovery, route maintenance and
error handling.  We demonstrate how {\awn} can be used to reason about
critical  protocol correctness properties by providing a detailed proof of loop
freedom. In contrast to evaluations using simulation or other formal methods such as model checking, 
our proof is generic and holds for any possible network scenario in terms of network topology, node mobility, traffic pattern, etc.
 A key contribution of this paper is the demonstration of how
the reasoning and proofs can relatively easily be adapted to protocol
variants.  
\end{abstract}
\hfuzz2pt 




\newpage
\section{Introduction}
\label{sec:intro}

Routing protocols are crucial to the dissemination of data packets between nodes in 
Wireless Mesh Networks (WMNs) and Mobile Ad Hoc Networks (MANETs). One of the most popular
routing protocols that is widely used in WMNs and MANETs is the Ad hoc On-Demand Distance
Vector (AODV) routing protocol~\cite{rfc3561}. AODV is one of the four protocols currently
standardised by the IETF MANET working group, and it also forms the basis of new WMN
routing protocols, including HWMP in the upcoming IEEE 802.11s wireless mesh network
standard~\cite{HWMP}. The details of the AODV protocol are laid out in the RFC 3561~\cite{rfc3561}. However, due to the use of English prose,
this specification contains ambiguities and contradictions. This can lead to significantly different implementations of the AODV routing protocol, depending on the developer's understanding and reading of the AODV RFC\@. In the worst case scenario, an AODV implementation may contain serious flaws, such as routing loops.

Traditional approaches to the analysis of AODV and many other AODV-based protocols~\nocite{DYMO22,HWMP,AODV-ST,SBM06,PPI08}[5,\,10,\,16,\,18,\,15]
 are simulation\linebreak[3] and test-bed experiments. While these are important and valid methods for protocol evaluation, in particular for quantitative performance evaluation, there are limitations in regards to the evaluation of basic protocol correctness properties. Experimental evaluation is resource intensive and time-consuming, and, even after a very long time of evaluation, only a finite set of network scenarios can be considered---no general guarantee can be given about correct protocol behaviour for a wide range of unpredictable deployment scenarios~\cite{Verisim}.  This problem is illustrated by recent discoveries of limitations in AODV-like protocols that have been under intense scrutiny over many years \cite{MK10}.  We believe that formal methods can help in this regard; they complement simulation and test-bed experiments as methods for protocol evaluation and verification, and provide stronger and more general assurances about protocol properties and behaviour.

This paper is based on a complete and accurate formal specification of the core functionality of the AODV routing protocol using the specification language \awn (Algebra of Wireless Networks)~\cite{ESOP12}. 
\awn provides the right level of abstraction to model key features such as 
unicast and broadcast, while abstracting from implementation-related details. As its semantics is completely unambiguous, specifying a protocol in such a framework enforces total precision and the removal of any ambiguities. A key contribution is to demonstrate how {\awn} can be used to
support reasoning about protocol behaviour and to provide a rigorous proof of key protocol
properties, using the example of loop freedom. In contrast to what can be achieved, e.g., by model checking, our proofs apply to all conceivable dynamic network topologies.

We analyse different readings of the AODV RFC, and show which
interpretations do satisfy the loop freedom criterion, and which do
not.  We also discuss two limitations of the AODV protocol and propose
solutions to them. We show how our formal specification can be used to
analyse the proposed modifications and show that these AODV variants
are loop free.

The rigorous protocol analysis discussed in this paper has the
potential to save a significant amount of time in the development and
evaluation of new network protocols, can provide increased levels of
assurance of protocol correctness, and complements simulation and
other experimental protocol evaluation approaches.

The remainder of this paper is organised as follows.
We briefly describe {\awn} in Section \ref{sec:language}, and use it to formally specify AODV in Section \ref{sec:formalspec}. We discuss one of several ambiguities and contradictions in the AODV RFC, and propose potential resolutions in Section \ref{sec:ambiguities}. We then summarise the key points of a detailed proof of loop freedom of AODV in Section \ref{sec:loopfreedom}, and demonstrate how the reasoning and proof can relatively easily be adapted to variants of the AODV protocol in Section \ref{sec:analysis}. We discuss related work in Section \ref{sec:related}, and summarise our work  
in Section \ref{sec:conclude}.
\section{AWN---An Algebra for Wireless Networks}
\label{sec:language}

\emph{Process algebras} are standard tools to describe interactions, communications and
synchronisations between a collection of independent agents, processes or network nodes.
They provide algebraic laws that allow formal reasoning.
For the specification of and for formal reasoning about AODV, we use
\awn~\cite{ESOP12,TR11}, a process algebra specifically tailored
for WMNs. \awn allows us to write a protocol specification in a simple
language, which makes it easy to read and to use.
Its key operators are {\em \conditional unicast}---allowing error handling in
response to failed communications while abstracting from link layer
implementations of the communication handling---and 
{\em local broadcast}---allowing a node to send messages to all 
its immediate neighbours as implemented by the physical and 
data link layer.

In this section we only give an overview of the main operations (Table~\ref{tb:procexpr}) and 
illustrate the use of {\awn} with a simple example. Additional
explanations and a full description can be found in~\cite{ESOP12,TR11}.

The example considers a network
of two nodes on which the same process is running.
One node broadcasts an integer value. A received message will be delivered to
the application layer if its value is $1$. Otherwise the node
decrements the value and broadcasts the new value. The behaviour of
each node can be modelled by:

\newcommand{\XP}{\keyw{X}}%
\newcommand{\YP}{\keyw{Y}}%
\begin{simpleProcess}
	\item[\XP(\keyw{n})]\hspace{-\labelsep}\ $\stackrel{{\it def}}{=}$
 	\textbf{broadcast}(\keyw{n}).\YP(\,)
	\item[\YP(\,)]\hspace{-\labelsep}\hspace{1.7mm} $\stackrel{{\it def}}{=}$
 	\textbf{receive}(\keyw{m}).([$\keyw{m} \mathord= 1$] \textbf{deliver}(\keyw{m}).\YP(\,)\ +\
		[$\keyw{m} \mathord{\not=} 1$] \XP($\keyw{m}\mathord-1$))
\end{simpleProcess}%

If a node is in a state $\XP(\keyw{n})$ it will broadcast $\keyw{n}$
and continue in state $\YP(\,)$. If a node is in state $\YP(\,)$, and it receives $\keyw{m}$, it has two
ways to continue. Process $[\keyw{m} \mathord=1]\,\textbf{deliver}(\keyw{m}).\YP(\,)$ is enabled if $\keyw{m} \mathord= 1$. In that case $\keyw{m}$
will be delivered to the application layer, and the process returns to $\YP(\,)$. Alternatively, if
$\keyw{m} \mathop{\not=} 1$, the process continues as $\XP(\keyw{m}\mathord-1)$. Note that calls to
processes use expressions as parameters, in this case $\keyw{m}\mathord-1$.

Assume that the nodes $A$ and $B$ are within communication range of each other; node~$A$ in state
$\XP(2)$, and node $B$ in $\YP(\,)$. 
Then, node $A$ broadcasts $2$ and
continues as $\YP(\,)$. Node $B$ receives $2$, and continues as $\XP(1)$. Next $B$
broadcasts $1$, and continues as $\YP(\,)$, while node $A$ receives $1$, and, since the condition
$\keyw{m}\mathord= 1$ is satisfied, \textbf{deliver}s $1$ and continues as $\YP(\,)$. 
This gives rise to transitions from one state to the other:
\vspace{-1mm}
\newcommand{\sm}[1]{\mbox{$\scriptstyle #1$}}
\[\begin{array}{r@{\;}l}
\XP(2)\,\|\,\YP(\,) \ar{{\sm{A}:\textbf{broadcast}\sm{(2)}}} \YP(\,)\,\|\,\XP(1)
&\ar{{\sm{B}:\textbf{broadcast}\sm{(1)}}}\\
&\ar{{\sm{A}:\textbf{deliver}\sm{(1)}}} \YP(\,)\,\|\,\YP(\,)
\end{array}\vspace{-1mm}\]
In state $\YP(\,)\,\|\,\YP(\,)$ no further activity is possible; the network has reached a \emph{deadlock}.

 \begin{table}[t]
  \vspace{1.2ex}
 \centering
{\small
  \setlength{\tabcolsep}{2.6pt}
 \begin{tabular}{|l|l|}
\hline
\rule[6.5pt]{0pt}{1pt}%
$X(\dexp{exp}_1,\ldots,\dexp{exp}_n)$& process name with arguments\\
$P+Q$ & choice between processes $P$ and $Q$\\
$\cond{\varphi}P$&conditional process;\\& execute $P$ only if condition $\varphi$ holds\\
$\assignment{\keyw{var}:=\dexp{exp}}P$&assignment followed by process $P$\\
$\broadcastP{\dexp{ms}}.P $&broadcast message \dexp{ms} followed by $P$\\
$\groupcastP{\dexp{dests}}{\dexp{ms}}.P$&iterative unicast to all destinations\\
	&\dexp{dests} (if broadcast is inappropriate)\\
$\unicast{\dexp{dest}}{\dexp{ms}}.P \prio Q$& unicast $\dexp{ms}$ to $\dexp{dest}$; if successful pro-\\
&ceed with $P$; otherwise with $Q$\\
$\deliver{\dexp{data}}.P$&deliver data to application layer\\
$\receive{\msg}.P$&receive a message\\
$P\|Q$		&parallel composition of nodes\\
\hline
\end{tabular}}

\caption{Process expressions}
\label{tb:procexpr}
\vspace{-4ex}
\end{table}

\section{A Formal Specification of AODV}
\label{sec:formalspec}

AODV is a reactive protocol, which means that routes are only established
on demand.  If a node $S$ wants to send a data packet to node $D$, but
currently does not know a route, it buffers the
packet and initiates a route discovery process
by broadcasting a route request (RREQ) message in the network. An intermediate
node $A$ that receives this message creates a \rte for a route
towards $S$, referred to as a \emph{reverse route}, and re-broadcasts the RREQ\@. This is repeated until the RREQ reaches the destination~$D$, or alternatively a node with a
route to $D$. In both cases, the node replies by unicasting a route reply (RREP) back to the source~$S$,
via the previously established reverse route. When forwarding RREP messages, a node
creates a \rte for node $D$, called the \emph{forward route}. When the RREP reaches~$S$, a route between $S$ and $D$  is established and data packets
can start to flow. In the event of link and route breaks, AODV uses route error messages (RERR) to
notify the affected nodes. AODV uses sequence numbers to indicate the freshness of routes and to
avoid routing loops. Full details are given in \cite{rfc3561}.

\subsection{Modelling AODV}\label{sec:aodv}
We present a model of
AODV using {\awn}. The formalisation
is a faithful representation of the core functionality of AODV as defined in~\cite{rfc3561}.
We currently do not model optional
features such as local route repair, expanding ring search, gratuitous route reply and multicast.
We also abstract from all timing issues, since
{\awn} currently does not support time. In concrete terms, this means that the
AODV timing parameters \verb|ACTIVE_ROUTE_TIMEOUT|, \verb|DELETE_PERIOD| and \verb|PATH_DISCOVERY_TIME|
are set to \mbox{infinity}.

\begin{table}[t]
\vspace{-1.5ex}
\centering
\begin{minipage}{0.999\columnwidth}
\renewcommand{\algindent}{0.1em}
  \renewcommand{\algorithmicelsif}{$+$\algorithmicif}
  \renewcommand{\algorithmiccomment}[1]{\textcolor{blue}{\hspace{\algindent}/*\,#The basic routine\,*/}}
  \algsetup{indent=0.7em}
  \algsetup{linenodelimiter=.,linenosize=\tiny}
  \begin{algorithm}[H]
    {\scriptsize
      \caption{\small The basic routine}
      \label{pro:aodv_short}
      \begin{algorithmic}[1]
        \DEFPROCESS{\aodvID}{\ip,\sn,\rt,\rreqs,\queues}
	\IFempty
		\receiveL{\msg}\ .																															\label{aodv:line2}
		\COMLINE{depending on the message, different processes are called}	
		\PAR
		\IF[new DATA packet]{$\msg = \newpkt{\data}{\dip}$}																	\label{aodv:line4}
			\pktP{\data}{\dip}{\ip}{\ip}{\sn}{\rt}{\rreqs}{\queues}																		\label{aodv:line5}
		\ELSIF[incoming DATA packet]{$\msg = \pkt{\data}{\dip}{\oip}$}  
			\pktP{\data}{\dip}{\oip}{\ip}{\sn}{\rt}{\rreqs}{\queues}
		\ELSIF[\!RREQ\!]{\!$\msg\mathop= \rreq{\hops}{\rreqid}{\dip}{\dsn}{\oip}{\osn}{\sip}$\!}					\label{aodv:line8}
			\COMLINE{update the route to \sip\ in \rt}																					\label{aodv:line9}
			\UPD{\rt:=\upd{\rt}{(\sip,0,\val,1,\sip,\emptyset)}}																		\label{aodv:line10}
			\rreqP{\hops}{\rreqid}{\dip}{\dsn}{\oip}{\osn}{\sip}{\ip}{\sn}{\rt}{\rreqs}{\queues}
		\ELSIF[RREP]{$\msg = \rrep{\hops}{\dip}{\dsn}{\oip}{\sip}$}															\label{aodv:line12}
			\COMLINE{update the route to \sip\ in \rt}																					\label{aodv:line13}
			\UPD{\rt:=\upd{\rt}{(\sip,0,\val,1,\sip,\emptyset)}}																		\label{aodv:line14}
			\rrepP{\hops}{\dip}{\dsn}{\oip}{\sip}{\ip}{\sn}{\rt}{\rreqs}{\queues}												\label{aodv:line15}
		\ELSIF[RERR]{$\msg = \rerr{\dests}{\sip}$}																					\label{aodv:line16}
			\COMLINE{update the route to \sip\ in \rt}																					\label{aodv:line17}
			\UPD{\rt:=\upd{\rt}{(\sip,0,\val,1,\sip,\emptyset)}}																		\label{aodv:line18}
			\rerrP{\dests}{\sip}{\ip}{\rt}{\sn}{\rreqs}{\queues}																			\label{aodv:line19}
		\ENDIFii
		\ENDPAR																																		\label{aodv:line20}
	\ELSIF[send a queued data packet]{$\mbox{Let } \dip\in\qD{\queues}\cap\akD{\rt}$}				\label{aodv:line22}
		\STATE \dots\label{aodv:lineDots}
	\ENDIFii

	\end{algorithmic}
    }
  \end{algorithm}

\end{minipage}
\vspace{-4ex}
\end{table}

In addition to modelling the complete set of core functionalities of the
AODV protocol, our model also covers the interface to higher protocol layers via the injection
and delivery of application layer data,
as well as the forwarding of {\datapackets} at intermediate nodes.
Although this is not part of the AODV protocol specification, it is necessary
for a practical model of any reactive routing protocol, where protocol activity is triggered via the
sending and forwarding of \datapackets.

Our AODV model consists of the following six processes:
\begin{itemize}
\item \textbf{$\AODV$}, the {main process}, reads a message
	from the message queue (Line~\ref{aodv:line2} of Process~\ref{pro:aodv_short})
	and calls the appropriate process {\PKT}, {\RREQ}, {\RREP}, or
	{\RERR} to handle it (Lines~\ref{aodv:line4}--\ref{aodv:line19}).
	The process also handles the forwarding of any queued {\datapackets}
	if a valid route to their destination is known (Lines~\ref{aodv:line22} ff.).
\item  \textbf{$\PKT$} deals with received \datapackets,
	including forwarding if a route exists, and sending an error message if the route is broken.
	If the data packet originates at the local node
	and no route to the destination exists, the
	process buffers the data packet and initiates a new
	route discovery process.
\item \textbf{$\RREQ$}  deals with received RREQ messages,
and will be discussed in detail below.
\item \textbf{$\RREP$} deals with received RREP messages,
including the updating of routing tables and handling of errors.%
\item \textbf{$\RERR$} models the processing of AODV error messages.
\item \textbf{$\QMSG$}  describes the general handling of incoming AODV messages: whenever
        a message is received, it is first stored in a FIFO queue. As soon as the
	corresponding node is able to handle a message it retrieves the oldest message from the
	queue and handles it.
\end{itemize}

{
\renewcommand{\ip}{\dval{ip}}
\renewcommand{\dip}{\dval{dip}}
\renewcommand{\oip}{\dval{oip}}
\renewcommand{\sip}{\dval{sip}}
\renewcommand{\rip}{\dval{rip}}
\renewcommand{\rt}{\dval{rt}}
\renewcommand{\r}{\dval{r}}
\renewcommand{\osn}{\dval{osn}}
\renewcommand{\dsn}{\dval{dsn}}
\renewcommand{\rsn}{\dval{rsn}}
\renewcommand{\flag}{\dval{flag}}
\renewcommand{\hops}{\dval{hops}}
\renewcommand{\nhip}{\dval{nhip}}
\renewcommand{\pre}{\dval{pre}}
\renewcommand{\dests}{\dval{dests}}
\renewcommand{\rreqid}{\dval{rreqid}}
\renewcommand{\rreqs}{\dval{rreqs}}

Each node in an AODV network maintains a \emph{routing table} to
  keep track of the node's routing information collected so far.
A routing table consists of sets of  entries of the form
$(\dip, \dsn, \flag, \hops, \nhip, \pre)$,  with $\dip$ being the
node identifier (typically IP address) of the ultimate destination node, and
$\dsn$ the destination sequence number, which represents the ``freshness''
of this \rte.
The $\flag$ parameter indicates whether an entry is valid or invalid, and $\hops$
represents the distance to the destination node $\dip$ in number of hops.
$\nhip$ identifies the next hop node along the route to node $\dip$, and $\pre$ is the set of
 \mbox{\emph{precursors}}---nodes
that ``rely" on this \rte for their own routes.
Following~\cite{rfc3561}, a \rte would also contain a {\em \dsnf}.
In the present paper we abstract from this flag, since (a) the main
results are independent of the existence of the flag, and (b) none
of the common implementations (AODV-UU~\cite{AODVUU},
Kernel-AODV~\cite{AODVNIST}, AODV-UIUC~\cite{Kawadia03},
AODV-UCSB~\cite{CB04}, AODV-ns2\,\footnote{\url{www.auto-nomos.de/ns2doku/aodv_8cc-source.html}}) maintains this flag.\footnote{Kernel-AODV implements the
  flag, but does not use it.}
Hence the specification here follows the implementations available.

In a routing table $\rt$ there is at most one entry for each destination
$\dip$; $\sqn{\rt}{\dip}$ denotes the sequence number of that entry
and $\status{\rt}{\dip}$, $\dhops{\rt}{\dip}$ and
  $\nhop{\rt}{\dip}$ its validity, hop count and next hop. 
Furthermore the sets $\kD{\rt}$ and $\akD{\rt}$ of destinations
contain all entries of $\rt$ for which there is an (arbitrary) entry or a valid entry, resp.
The function $\fnupd$ updates a routing table $\rt$ with an entry $r$, which is one
of the major activities of $\AODV$:

\vspace{-\abovedisplayskip}
{\small
\[\begin{array}{@{}r@{\hspace{0.1em}}c@{\hspace{0.1em}}l@{}}
\upd{\rt}{\r}&:=& \left\{
\begin{array}{@{}l@{\hspace{3pt}}l@{}r@{}}
\phantom{n}\rt\mathop\cup\{\r\} & \mbox{if }  \pi_{1}(\r)\not\in\kD{\rt}&\textit{//r is new}\\[1mm]
\nrt\mathop\cup\{\nr\}&\mbox{if } \sqn{\rt}{\pi_{1}(\r)}\mathord{<}\pi_{2}(\r)&\textit{//fresher} \\[1mm]
\nrt\mathop\cup\{\nr\}&\mbox{if } \sqn{\rt}{\pi_{1}(\r)}\mathord{=}\pi_{2}(\r) \\&\phantom{\mbox{if}}\wedge \dhops{\rt}{\pi_{1}(\r)}\mathord{>}\pi_{4}(\r)\!&\textit{//shorter} \\[1mm]
\nrt\mathop\cup\{\nr\}&\mbox{if } \sqn{\rt}{\pi_{1}(\r)}\mathord{=}\pi_{2}(\r)&\textit{//replaces}
\\&\phantom{\mbox{if}}\wedge \status{\rt}{\pi_{1}(\r)}\mathord{=}\inval&\textit{ invalid} \\[1mm]
\nrt\mathop\cup\{\nr'\}&\mbox{if } \pi_{2}(\r)\mathord{=}0&\hspace{-4mm}\textit{//unk.\ sqn}\\[1mm]
\nrt\mathop\cup\{\ns\}&\mbox{otherwise\ ,}
\end{array}
\right.
\end{array}\]
}
\vspace{-\abovedisplayskip}

\noindent
where the projections $\pi_{1}$, $\pi_{2}$ and $\pi_{4}$ select the
respective component from an entry, namely the destination,
the destination sequence number and the hop count. $\s$ is the current entry in $\rt$
for destination $\pi_{1}(\r)$ (if it exists); and $\nrt\mathop{:=} \rt-\s$
removes $\s$ from $\rt$.  The entry $\nr$
is identical to $\r$ except that the precursors from the corresponding
\rte are added and $\ns$
is generated from $\s$ by adding the precursors of $\r$. The entry
$\nr'\!$ is identical to $\nr$ except that the sequence number is replaced by the one from the routing table (route $s$).

If a route is not valid any longer, instead of deleting it, AODV sets
its validity flag to invalid. This way, the stored information on the
route, such as the sequence number and hop count, remains accessible.
We model route invalidation by a function {\fninv} whose
arguments are a routing table and a set $\dests$ of pairs $(\rip,\rsn)$
of a destination {\rip} to be invalidated, and the sequence
number of the invalidated \rte.
Normally, $\rsn$ is obtained by incrementing the last known sequence
number of the route.

In our formalisation, a route request message has the form
\rreq{\hops}{\rreqid}{\dip}{\dsn}{\oip}{\osn}{\sip},
where {\hops} is the number of hops the RREQ has already travelled from its origin {\oip},
and {\rreqid} (in combination with {\oip}) is a unique identifier of the message.
{\dip} is the destination node identifier (IP address) of the route request and
{\dsn} the last known corresponding sequence number.
The parameter {\oip} is the address of the originator of the route request and
{\osn} is its sequence number.
Finally, {\sip} represents the sender IP address, i.e., the address
of the intermediate node from which the request was received.
Any node forwarding such a message updates {\sip} with its own
address, increments {\hops}, and retains all other parameters.
A reply to such a message has the form
\rrep{\hops}{\dip}{\dsn}{\oip}{\sip}, where {\dip} and {\oip} are
copied from the corresponding RREQ message and {\hops} is the distance
from {\dip} to {\sip}.
}
The processes {\RREQ} and {\RREP} that handle incoming RREQ and RREP
messages maintain variables {\dip}, {\oip}, etc.\ to store the values
of the parameters of these messages, as summarised
below.

{\small\begin{center}
\setlength{\tabcolsep}{3pt}
\begin{tabular}{|c|l|}
\hline
\textbf{Variables} & \textbf{Used for}\\
\hline
\ip   &  address of current node
\\
\dip  &  destination address
\\
\oip  &  originator of a route request or data packet
\\
\rip  &  destination of invalid route
\\
\sip  &  sender of AODV control message
\\
\nhip &  next hop towards some destination
\\
\hline
\end{tabular}
\end{center}
}

The process {\aodvID}, specified by
Process~\ref{pro:aodv_short}, deals with the message handling of the node.
It stores its own address in the variable {\ip}, its own sequence number in {\sn}, manages its routing table
$\rt$, records all route requests seen so far in $\rreqs$ and
maintains in $\queues$ {\datapackets} to be sent. Initially, $\rt$, $\rreqs$ and
$\queues$ are set to empty, and $\sn$ to $1$.

\subsection{Route Request Handling}

{
\renewcommand{\ip}{\dval{ip}}
\renewcommand{\dip}{\dval{dip}}
\renewcommand{\oip}{\dval{oip}}
\renewcommand{\sip}{\dval{sip}}
\renewcommand{\rip}{\dval{rip}}
\renewcommand{\rt}{\dval{rt}}
\renewcommand{\r}{\dval{r}}
\renewcommand{\osn}{\dval{osn}}
\renewcommand{\dsn}{\dval{dsn}}
\renewcommand{\rsn}{\dval{rsn}}
\renewcommand{\flag}{\dval{flag}}
\renewcommand{\hops}{\dval{hops}}
\renewcommand{\nhip}{\dval{nhip}}
\renewcommand{\pre}{\dval{pre}}
\renewcommand{\dests}{\dval{dests}}
\renewcommand{\rreqid}{\dval{rreqid}}
\renewcommand{\rreqs}{\dval{rreqs}}

In this paper, we discuss only  the model of the RREQ process; 
see \cite{TR11} for a complete model
 of all AODV processes.\footnote{There, the \dsnf is modelled as well.}

A route discovery in AODV is initiated by a source node broadcasting a RREQ message.
Process~\ref{pro:rreq_short} shows our process algebra specification of the handling of a RREQ message received
by a node \ip.

If the RREQ with the same {\oip} and {\rreqid} has been seen previously by the node, it is
ignored, and we go back to the main AODV process (Lines~\ref{rreq:line2}--\ref{rreq:line3}). If the
RREQ is new (Line~\ref{rreq:line4}), we update the routing table by adding a ``reverse route'' entry
to {\oip}, the originator of the RREQ, via node {\sip}, with distance \hops+1 (Line~\ref{rreq:line6}).  If there
already is a route to {\oip} in the node's routing table {\rt}, it is only updated with the new
route, if the new route is ``better'', i.e., fresher and/or shorter
and/or replacing an invalid route
(cf.\ Section~\ref{sec:aodv}). The process also adds the message to the list of known RREQs (Line~\ref{rreq:line8}).

Lines~\ref{rreq:line10}--\ref{rreq:line19} deal with the case where the node receiving the RREQ is the
intended destination, i.e., {\dip}={\ip} (Line~\ref{rreq:line10}).
In this case, a RREP message needs to be sent to the originating node \oip. According to the AODV RFC,
the node's sequence number is set to the maximum of the node's current sequence number and the
destination sequence number (\dsn) in the RREQ message (Line~\ref{rreq:line12}).

The RREP message is initialised as follows: hop count (\hops) is set to $0$, the destination (\dip) and originator (\oip)
are copied from the corresponding RREQ message and the destination's sequence number is the node's sequence number \sn. Of course, the sender's IP address (\sip) is set to the node's ip (Line~\ref{rreq:line14a}).
The RREP message is unicast to the next hop along the reverse route back to the originator of the
corresponding RREQ message, and if this is successful, the process goes back to the AODV process (Line~\ref{rreq:line14}).

If the unicast of the RREP fails, we proceed with Lines~\ref{rreq:line15}--\ref{rreq:line19}, in which a route error
(RERR) message is generated and sent.  This conditional unicast is implemented in our model with the \awn construct 
$\unicast{\dexp{dest}}{\dexp{ms}}.P \prio Q$ (Lines~\ref{rreq:line14a}ff.)
}
\begin{table}[t]%
\vspace{-0.7em}%
\centering%
\begin{minipage}{0.999\columnwidth}%
\renewcommand{\algindent}{0.1em}
  \renewcommand{\algorithmicelsif}{$+$\algorithmicif}
  \renewcommand{\algorithmiccomment}[1]{\textcolor{blue}{\hspace{\algindent}/*\,#RREQ handling\,*/}}
  \algsetup{indent=0.7em}
  \algsetup{linenodelimiter=.,linenosize=\tiny}
  \begin{algorithm}[H]
    {\scriptsize
      \caption{\small RREQ handling}
      \label{pro:rreq_short}
      \begin{algorithmic}[1]
        \DEFPROCESS{\RREQ}{\hops,\rreqid,\dip,\dsn,\oip,\osn,\sip\,,\,\ip,\sn,\rt,\rreqs,\queues}
	\IF[the RREQ has been handled before]{$(\oip,\rreqid)\in\rreqs$}																\label{rreq:line2}
		\aodvL{\ip}{\sn}{\rt}{\rreqs}{\queues} \COM{silently ignore RREQ}															\label{rreq:line3}
	\ELSIF[the RREQ is new to this node]{$(\oip,\rreqid)\not\in\rreqs$}															\label{rreq:line4}
		\COMLINE{update the route to \oip\ in \rt}																								\label{rreq:line5}
		\UPD{\rt:=\upd{\rt}{(\oip,\osn,\val,\hops+1,\sip,\emptyset)}}																		\label{rreq:line6}
		\COMLINE{update \rreqs\ by adding $(\oip,\rreqid)$}																			\label{rreq:line7}
		\UPD{\rreqs:=\rreqs\cup\{(\oip,\rreqid)\}}																								\label{rreq:line8}
		\PAR	
		\IF[this node is the destination node]{$\dip=\ip$}																					\label{rreq:line10}
			\UPD{\sn:=\max(\sn,\dsn)}	\COMMENT{update the sqn of \ip}															\label{rreq:line12}
			\COMLINE{unicast a RREP towards \oip\ of the RREQ}																	\label{rreq:line13}
			\STARTPRIO
					\unicast{\nhop{\rt}{\oip}}{{\rrep{$0$}{\dip}{\sn}{\oip}{\ip}}}\,. 														\label{rreq:line14a}								
					\aodvL{\ip}{\sn}{\rt}{\rreqs}{\queues}																							\label{rreq:line14}
				\PRIO
					\COMspec{if transmission fails, a RERR is generated}																\label{rreq:line15}
					\UPD{\dests:=\{(\rip,\inc{\sqn{\rt}{\rip}})|\rip\in\akD{\rt}\wedge \newline
							\hspace*{15mm}\nhop{\rt}{\rip}=\nhop{\rt}{\oip}\}}																\label{rreq:line16}
					\UPD{\rt:=\inv{\rt}{\dests}}																											\label{rreq:line18a}
					\UPD{\pre:=\bigcup\{\precs{\rt}{\rip}\,|\,(\rip,*)\in\dests\}}																\label{rreq:line17}
					\UPD{\dests:=\{(\rip,\rsn)\,|\,(\rip,\rsn)\in\dests\wedge \newline
							\hspace*{15mm}\precs{\rt}{\rip}\not=\emptyset}						\label{rreq:line17a}
					\groupcast{\pre}{\rerr{\dests}{\ip}}\,.\label{rreq:line18b}
                                        \aodvL{\ip}{\sn}{\rt}{\rreqs}{\queues}									\label{rreq:line19}
				\ENDPRIO		
	\ELSIF[this node is not the destination node]{$\dip\not=\ip$}																		\label{rreq:line20}			
			\PAR
\COMLINE{valid route to \dip\ that is fresh enough}
			\IF{$\dip\in\akD{\rt} \wedge \dsn \leq  \sqn{\rt}{\dip} \wedge\sqn{\rt}{\dip}\not=0$}								\label{rreq:line22}
					\COMLINE{update \rt\ by adding precursors}																				\label{rreq:line23}
					\UPD{\rt := \addprecrt{\rt}{\dip}{\{\sip\}}}																						\label{rreq:line24}
					\UPD{\rt := \addprecrt{\rt}{\oip}{\{\nhop{\rt}{\dip}\}}}																		\label{rreq:line25}
				\COMLINE{unicast a RREP towards the \oip\ of the RREQ}
				\STARTPRIO
					\unicast{\nhop{\rt}{\oip}}{\newline\hspace*{8.1mm}\rrep{\dhops{\rt}{\dip}}{\dip}{\sqn{\rt}{\dip}}{\oip}{\ip}}\,.						\label{rreq:line26}
					\aodvL{\ip}{\sn}{\rt}{\rreqs}{\queues}
				\PRIO
					\COMspec{if transmission fails, a RERR is generated}	
					\UPD{\dests:=\{(\rip,\inc{\sqn{\rt}{\rip}})\,|\,\rip\in\akD{\rt}\wedge \newline
							\hspace*{15mm}\nhop{\rt}{\rip}=\nhop{\rt}{\oip}\}}																\label{rreq:line28}
					\UPD{\rt:=\inv{\rt}{\dests}}																											\label{rreq:line30}					
					\UPD{\pre:=\bigcup\{\precs{\rt}{\rip}\,|\,(\rip,*)\in\dests\}}																\label{rreq:line29}
					\UPD{\dests:=\{(\rip,\rsn)\,|\,(\rip,\rsn)\in\dests\wedge \newline
							\hspace*{15mm} \precs{\rt}{\rip}\not=\emptyset}																	\label{rreq:line29a}
					\groupcast{\pre}{\rerr{\dests}{\ip}}\,. 																							\label{rreq:line31}
					\aodvL{\ip}{\sn}{\rt}{\rreqs}{\queues}
				\ENDPRIO
			\COMspeci{no valid route that is fresh enough}
			\ELSIF{$\dip\not\in\akD{\rt} \vee \sqn{\rt}{\dip} <  \dsn \vee\sqn{\rt}{\dip}=0$}										\label{rreq:line32}
				\COMLINE{no further update of \rt}
				\broadcast{\rreq{$\hops+1$}{\rreqid}{\dip}{$\newline\hspace*{20.4mm}\max(\sqn{\rt}{\dip},\dsn)$}{\oip}{\osn}{\ip}}\,.						\label{rreq:line34}
				\aodvL{\ip}{\sn}{\rt}{\rreqs}{\queues}	\label{rreq:line35}
			\ENDIFii
			\ENDPAR
		\ENDIFii
		\ENDPAR
	\ENDIFii

	\end{algorithmic}
    }
  \end{algorithm}
\end{minipage}%
\vspace{-5ex}%
\end{table}%
{
\renewcommand{\ip}{\dval{ip}}%
\renewcommand{\dip}{\dval{dip}}%
\renewcommand{\oip}{\dval{oip}}%
\renewcommand{\sip}{\dval{sip}}%
\renewcommand{\rip}{\dval{rip}}%
\renewcommand{\rt}{\dval{rt}}%
\renewcommand{\r}{\dval{r}}%
\renewcommand{\osn}{\dval{osn}}%
\renewcommand{\dsn}{\dval{dsn}}%
\renewcommand{\rsn}{\dval{rsn}}%
\renewcommand{\flag}{\dval{flag}}%
\renewcommand{\hops}{\dval{hops}}%
\renewcommand{\nhip}{\dval{nhip}}%
\renewcommand{\pre}{\dval{pre}}%
\renewcommand{\dests}{\dval{dests}}%
\renewcommand{\rreqid}{\dval{rreqid}}%
\renewcommand{\rreqs}{\dval{rreqs}}%
We assume that, as is the case for relevant wireless
technologies such as IEEE 802.11, unicast messages are acknowledged, and we therefore can determine whether
the transmission was unsuccessful and the link to the next node towards {\oip} is broken.
In this case, the node sends a RERR message to all nodes that rely on the broken link for one of their routes.
For this, we first determine which destination nodes are affected by the broken link, i.e., the nodes
that have this unreachable node listed as a next hop in the routing table (Line~\ref{rreq:line16}).
Here, the operator $\fninc$ increments the sequence numbers of those entries.
Then, we invalidate any affected \rtes (Line~\ref{rreq:line18a}), and determine the list of
\emph{precursors}, which are the
neighbouring nodes that have a
route to one of the affected destination nodes via the broken link (Line~\ref{rreq:line17}).
Finally, using the \awn \textbf{groupcast} primitive, a RERR message is sent via unicast to all these precursors
(Line~\ref{rreq:line18b}),\pagebreak[2]
listing only those invalidated destinations with a non-empty set of
precursors (Line~\ref{rreq:line17a}).

Lines~\ref{rreq:line20}--\ref{rreq:line35} deal with the case where the node receiving the RREQ is
not the destination, i.e., ${\dip}\mathop{\neq}{\ip}$ (Line~\ref{rreq:line20}).  The node can respond to the RREQ with a
corresponding RREP on behalf of the destination node \dip, if its route to {\dip} is ``fresh enough''
(Line~\ref{rreq:line22}).  This means that (a) the node has a valid route to \dip, (b) the destination sequence
number in the node's current \rte $(\sqn{\rt}{\dip})$ is greater than or equal to the
requested sequence number to {\dip} in the RREQ message, and (c) the sequence number is valid, i.e., it
is not unknown ($\sqn{\rt}{\dip} \mathord{\neq} 0$).  If these three conditions are met
(Line~\ref{rreq:line22}), the node generates a RREP message and unicasts it back to the originator
node {\oip} via the reverse route. To this end, it copies the sequence number for the
destination $\dip$ from the routing table $\rt$ into the destination sequence number field of the RREP message and it places its distance in hops from the destination ($\dhops{\rt}{\dip}$) in the corresponding field of the new reply
(Line~\ref{rreq:line26}).
 As usual, the unicast might fail, which
causes the same error handling (Lines~\ref{rreq:line28}--\ref{rreq:line29a}).
Just before unicasting the RREP message, the intermediate node updates the forward {\rte} to
$\dip$ by placing the last hop node ($\sip$)
into the precursor list for
that entry (Line~\ref{rreq:line24}).
Likewise, it updates the reverse {\rte} to {\oip} by placing the first hop $\nhop{\rt}{\dip}$
towards $\dip$ in the precursor list for that entry
(Line~\ref{rreq:line25}).

If the node is not the destination and there is either no route to the
destination $\dip$ inside the routing table or the route is not fresh enough,
the route request received has to be forwarded. This happens in Line~\ref{rreq:line34}.
The information inside the forwarded request is mostly copied from the request received.
Only the hop count is increased by $1$ and the destination sequence number is set
to the maximum of  destination sequence number in the RREQ packet
and the the current sequence number for $\dip$ in the routing table.
In case $\dip$ is an unknown destination, $\sqn{\rt}{\dip}$ returns the
unknown sequence number $0$.
}

\section{Ambiguities in the RFC}\label{sec:ambiguities}

The formal specification of AODV, outlined above and given in full detail in \cite{TR11},
closely follows the RFC 3561~\cite{rfc3561}, the official specification of the protocol.
However, the RFC contains several ambiguities and contradictions; an inventory is
presented in \cite{TR11}, and for each ambiguity or contradiction a number of ways to
resolve them is listed.  An \emph{interpretation} of the RFC is given by the allocation of a
resolution to each of the ambiguities and contradictions. Each reading, implementation,
or formal analysis of AODV must pertain to one of its interpretations. The formal
specification of AODV in \cite{TR11} constitutes one interpretation; the inventory
of ambiguities and contradictions is formalised by specifying each resolution
of each of the ambiguities and contradictions as a modification of this formal
specification, typically involving a rewrite of a few lines of code only.

A crucial contradiction in the RFC concerns the question of what would happen if a node has a
valid \rte for a destination $D$, with destination sequence number $n$, and an error
message is received from the next hop towards $D$, saying that the route to $D$ is broken,
and stating for this route a destination sequence number $m$, which may be smaller than $n$.
Section 6.11 of the RFC unambiguously states that in such a case the node updates its \rte
to $D$ by marking the route as invalid, and {\em copying} the destination sequence number from
the incoming route error message. However, Section 6.1 of the RFC states that if $m<n$,
any information related to $D$ in the AODV message must be discarded.

One can show \cite{TR11} that in case no node will ever store a \rte to itself (a self-entry), the above
situation will never occur. However, the RFC does not explicitly exclude self-entries\footnote{The Kernel-AODV, AODV-UIUC, AODV-UCSB and\newline AODV-ns2 implementations allow self-entries to occur.}, and
they can in fact occur \cite{TR11} in response to the standard handling of RREP messages.

The following ways to resolve this contradiction have been listed in \cite{TR11}: 
\begin{enumerate}[(a)]
\item\label{it:amb_i}
	Follow Section 6.11 of the RFC, in defiance of 6.1, i.e., {\em always\/} invalidate the \rte, and copy the destination sequence number from the error message to the corresponding entry in the routing table.%
	\footnote{It could be argued that this is not a reasonable interpretation of the RFC, since Section\ 6.1 should have priority over 6.11. However, this priority is not explicitly stated.}
\item\label{it:amb_ii}
	 Follow Section 6.11 only where it does not contradict 6.1, i.e., invalidate the \rte and copy the destination sequence number {\em only\/} if $m\geq n$.
\item\label{it:amb_iii}
	Always invalidate the \rte, but update the destination sequence number to $\max(m,n)$. 

\item\label{it:amb_iv}
	Always invalidate the \rte, but update the destination sequence number to $\max(m,n\mathord+1)$.
\item\label{it:amb_v}
	Invalidate the \rte and update the destination sequence number to $\max(m,n\mathord+1)$ only if $m\geq n$.%
	\footnote{The case $\max(m,n)$ if $m\geq n$ need no separate consideration, since it is equivalent to \eqref{it:amb_ii}.}
\item\label{it:amb_vi}
	Invalidate the \rte only if $m>n$.%
	\footnote{Here, it does not matter whether we copy, take $\max(m,n)$ or $\max(m,n\mathord+1)$; they are all equivalent.}
\item\label{it:amb_vii}
	Forbid self-entries; if an incoming RREP message would create a self-entry, discard that message.
\item\label{it:amb_viii}
	Forbid self-entries; if an incoming RREP message would create a self-entry, forward that message without updating the node's routing table.
\end{enumerate}
It should be noted that only resolutions \eqref{it:amb_i} and \eqref{it:amb_ii} are compliant with the RFC\@.
However, in \cite{ICNP12} we have shown that any interpretation based on resolutions \eqref{it:amb_i} or \eqref{it:amb_ii} gives rise to routing loops, so in order to arrive at a loop-free version of
AODV, one has to deviate from the RFC\@. Here, as in \cite{TR11}, we do so by choosing
resolution \eqref{it:amb_vi}.

The above is only one of many ambiguities; another one  is presented in Section~\ref{ssec:interpretation}.

\section{Loop freedom}
\label{sec:loopfreedom}

We now formalise loop freedom and sketch a proof that our
detailed specification of AODV cannot create routing loops.
We also show how such a formal proof can form a
baseline for evaluating variants of AODV\@---some of them will be loop free, others can yield loops. 

First we formalise what it means for the routing tables
established by AODV (our specification) to be free of loops.  Let $\IP$ be the set of network nodes and
$\dval{dip}\mathop{\in} \IP$ a particular destination; let $N$ be a
state of the network, encompassing
the current values of all variables maintained by all nodes.

\pagebreak
\mbox{}\vspace{-5.5ex} 

\noindent
The \emph{routing graph} $\RG{N}{\dval{dip}}$ for destination \dval{dip} in state $N$ is the
directed graph $(\IP,E)$ with set of vertices $\IP$ and set of edges $E \subseteq \IP{\times}\IP$
consisting of the pairs $({\dval{ip}},{\dval{ip}'})$ such that
$\dval{ip}\mathop{\not=}\dval{dip}$ and
$(\dval{dip},*,\val,*,\dval{ip}'\!,*)$
occurs in the routing table of \dval{ip} in state $N$.
Thus, there is an  edge $(\dval{ip}, \dval{ip}')$ if node \dval{ip} is not the destination
\dval{dip}, but has a valid entry for \dval{dip}, and $\dval{ip}'$ is the
next hop according to that entry.
Loops in directed graphs are defined to be paths following edges which return to a vertex.

A (network) state $N$ is \emph{loop free} if the
routing graphs $\RG{N}{\dval{dip}}$ are loop free for all
$\dval{dip}\mathop{\in}\IP$.  The specification of AODV is
\emph{loop free} iff all reachable
states are loop free.

Let us now turn towards a proof of loop freedom for AODV$\!$\@.
It relies on a number of
\emph{invariants}---statements that hold for all reachable
states of our model.  An invariant
is usually verified by showing that it holds
for all possible initial states, and that, for any transition
\plat{$N\ar{\ell}N'$} derived by our operational semantics
\cite{TR11}, if it holds for state $N$ then it also holds for state $N'\!$, reached after performing some action $\ell$.
These transitions can be traced back to the line numbers in our process
declarations {\AODV}, {\PKT}, {\RREQ}, {\RREP}, {\RERR} and {\QMSG}.

A proof of AODV's loop freedom using invariants has first been proposed in \cite{BOG02}.\footnote{In fact, the
  same idea occurs already in \cite{AODV99}, but without the formalisation in terms of
  invariants.  However, that proof fails to consider some cases that do
  occur in AODV and might yield routing loops \cite{ICNP12}.}
The main invariant of~\cite{BOG02}  states that
\emph{if node \dval{ip} has a \rte for destination \dval{dip} with next hop \dval{nhip},
then also node \dval{nhip} has a routing table entry for \dval{dip}, and the latter has a
larger destination sequence number, or an equal one with a strictly smaller hop count.}

This invariant is claimed to hold regardless whether the \rtes for \dval{dip} at \dval{ip}
and \dval{nhip} are marked as valid or invalid. Nevertheless, the following example shows
that it does not hold for the current version of AODV\@.

\begin{figure}[h]
\centerline{
\includegraphics[scale=1]{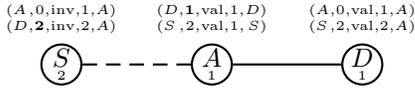}
}
\caption{Sequence numbers can go down\protect\footnotemark}
\label{fg:sundown}
\end{figure}
\footnotetext{We omit the precursor set; the
\rtes in each node are represented as $(\dval{dip},
\dval{dsn}, \dval{flag},\dval{hops}, \dval{nhip})$, described in Section~\ref{sec:aodv}.}

The network state depicted in Figure~\ref{fg:sundown} can occur when, after a standard RREQ-RREP cycle, a link break is detected. 
We assume a simple linear topology of $3$ nodes. Below the node names, the circles list the nodes' own sequence numbers, which we assume to be $1$ at the initial state.
The example starts with empty routing tables and node $S$ searching for a route to node $D$. 
Before broadcasting a RREQ message, node $S$ increments its sequence number by $1$. 
Due to the successful exchange of RREQ-RREP messages, all routing tables are updated.
After that, the link between nodes $S$ and $A$ goes down, $S$ detects the link break
and updates its routing table: it sets all entries in its
 routing table with next hop $A$ to {\em invalid} and increases the
destination sequence number of the route to $D$ to $2$; the destination
sequence number in the route to $A$ is unknown ($0$) and hence not incremented. Now the
destination sequence number on the route to node $D$ is $2$ in the routing table of
$S$ and $1$ inside $A$'s routing table, hence this number can  go down.

One way to avoid this problem is to claim the invariant only for the case that the \rtes
at \dval{ip} and \dval{nhip} are both marked as valid. This is what we do in
Theorem~\ref{thm:inv_a} below, and it suffices to obtain loop freedom of AODV\@.
However, in order to prove Theorem~\ref{thm:inv_a}, we 
need an invariant that also
takes invalid \rtes into account (cf.\ Proposition~\ref{prop:inv_nsqn}), so it is not
possible to avoid the above problem altogether.

\newcommand{\hopsc}{\dval{hops}_c}
\newcommand{\oipc}{\dval{oip}_{\hspace{-1pt}c}}
\newcommand{\rreqidc}{\dval{rreqid}_{c}}
\newcommand{\osnc}{\dval{osn}_c}
\newcommand{\ipc}{\dval{ip}_{\hspace{-1pt}c}}
\renewcommand{\kno}{\keyw{kno}}
\renewcommand{\unkno}{\keyw{unk}}
\newcommand{\destsc}{\dval{dests}_c}
\newcommand{\Prop}[1]{Proposition~\ref{prop:#1}}
\newcommand{\Eq}[1]{\eqref{eq:#1}}
\renewcommand{\kd}[2][N]{\fnkD_{#1}^{\ensuremath{#2}}}
\renewcommand{\akd}[2][N]{\fnakD_{#1}^{\ensuremath{#2}}}
\newcommand{\decremented}{\mathbin{\stackrel{\bullet}{\raisebox{0pt}[2pt]{$-$}}}1}
To compensate for the increase of a sequence number in case of route invalidation, we introduce the concept of  a \emph{net sequence number} of a route  to
$\dval{dip}$ according to the routing table of node $\dval{ip}$ in state
$N$, which combines ``freshness'' and validity:\vspace{-3ex}%

{\small
\[\begin{array}{@{}r@{\hspace{0.1em}}c@{\hspace{0.1em}}l@{}}
\nsq{\dval{ip}}&:=&\left\{
\begin{array}{@{\hspace{-0.1em}}l@{}}
\plat{$\sq{\dval{ip}}$}\ \ \mbox{if }\plat{$\sta{\dval{ip}}$}\mathop=\val \mathop\vee \sq{\dval{ip}}\mathop=0\\[1mm]
\plat{$\sq{\dval{ip}}$}\mathord-1\ \ \mbox{otherwise}\ .
\end{array}
\right.
\end{array}\]
}%
Here, we write $\sq{\dval{ip}}$ for \plat{$\sqn{\dval{rt}}{\dip}$} in case \dval{rt}
happens to be the routing table maintained by a node with IP address \dval{ip} in state $N$ of the
network. Likewise \plat{$\sta{\dval{ip}}$}
denotes the validity of the route from \dval{ip} to \dval{dip} according to the routing table of
\dval{ip} in state $N$, \plat{$\dhp{\dval{ip}}$} its hop count, and \plat{$\nhp{\dval{ip}}$} its
next hop.  Furthermore $\D{\dval{ip}}$, abbreviating $\kD{\dval{rt}}$, is the set of destinations for which there is
a valid entry in the routing table of \dval{ip}.

In this section we state the key theorems and sketch some of the proofs;
all details can be found in~\cite{TR11}. In particular, we show
only proofs w.r.t.\ Process~\ref{pro:rreq_short}, the RREQ handling,
and the displayed portion of Process~\ref{pro:aodv_short}.

\begin{prop}\label{prop:starcast_ii}
If a route request is sent (forwarded) by a node $\ipc$ different from
the originator of the request then the content of $\ipc$'s routing
table must be fresher or at least as good as the information inside
the message.
\vspace{-1mm}
\[
\begin{array}{@{\hspace{-4mm}}r@{}c@{\hspace{0.3em}}l@{\hspace{-4mm}}}
	  &&N\ar{R:\textbf{broadcast}(\rreq{\hopsc}{*}{*}{*}{\oipc}{\osnc}{\ipc})}N'\wedge\ipc\mathop{\not=}\oipc\\
	  &\Rightarrow&
	  \oipc\in\D{\ipc}
	  \wedge\big(\sq[\oipc]{\ipc}\mathop>\osnc
	  \vee (\sq[\oipc]{\ipc}\mathop=\osnc\\
	  &&\ans \dhp[\oipc]{\ipc}\mathop\leq\hopsc\wedge \sta[\oipc]{\ipc}\mathop=\val)\big)
\end{array}
\]\end{prop}
The next proposition says that in the routing table of a node \dval{ip},
an entry to a destination \dval{dip} will never be deleted, and
the net sequence number of the entry will never go down.
\begin{prop}\label{prop:7.23}
Let $\dval{ip},\dval{dip}\mathbin\in\IP$, and assume $N \ar{\ell}N'\!$, i.e.\ the network
proceeds from state $N$ to state $N'$ by the occurrence of some action $\ell$.
If $\dval{dip}\in\kd{\dval{ip}}$
then $\dval{dip}\in\kd[N']{\dval{ip}}$
and $\nsq{\dval{ip}}\leq \fnnsqn_{N'}^{\dval{ip}}(\dval{dip})$.
\end{prop}
\begin{proof}
In our model of AODV, the only way to change a routing table is by means of the operations
$\fnupd$, $\fnaddprecrt$ and $\fninv$. None of them ever deletes an entry altogether, or
decreases its net sequence number.
\end{proof}
In this proof it is essential that we use resolution~\eqref{it:amb_vi} of
the ambiguity in the RFC presented in Section~\ref{sec:ambiguities}.
\Prop{7.23} would not hold under resolutions (a), (b) or (c).

\begin{prop}\rm
\label{prop:inv_nsqn}
If, in a reachable network state $N\!$, a node $\dval{ip}\mathop\in\IP$ has a
routing table entry to $\dval{dip}$, then also the next hop
\dval{nhip} towards \dval{dip}, if not \dval{dip} itself, has a
routing table entry to $\dval{dip}$, and the net sequence number of
the latter entry is at least as large as that of the former.
\[
\begin{array}{@{\hspace{-4mm}}r@{}c@{\hspace{0.3em}}l@{\hspace{-4mm}}}
	  &&\dval{dip}\in\D{\dval{ip}}\ans\dval{nhip}\not=\dval{dip}\\
&\Rightarrow&\dval{dip}\in\D{\dval{nhip}} \ans \nsq{\dval{ip}}\leq \nsq{\dval{nhip}}\ ,
\end{array}\]
where $\dval{nhip}:=\nhp{\dval{ip}}$ is the IP address of the next hop.
\end{prop}

\begin{proof}
In the initial state, the invariant holds  since all routing tables are empty (cf.\ Section~\ref{sec:aodv}).
Next we assume that the property holds and check each line in
Pro.~\ref{pro:aodv_short} and
Pro.~\ref{pro:rreq_short} which could invalidate it.

A modification of the routing table of \dval{nhip} is harmless, as it
can only increase $\kd{\dval{nhip}}$
as well as $\nsq{\dval{nhip}}$ (cf.\ \Prop{7.23}).
  
Adding precursors to routes of \dval{ip} does not harm since the 
invariant does not depend on precursors.
It remains to examine all calls of $\fnupd$ and $\fninv$ to the routing table of \dval{ip}.
Without loss of generality we restrict attention to those applications of $\fnupd$
or $\fninv$ that actually modify the entry for \dval{dip}, beyond its
precursors; if $\fnupd$ only adds some precursors in the routing
table, the invariant---which is assumed to hold before---is maintained.

\begin{description}
  \itemsep=0mm
\item[Pro.~\ref{pro:aodv_short}, Lines~\ref{aodv:line10}, \ref{aodv:line14}, \ref{aodv:line18}:]
	The entry $(\sip,0,\val,1,\sip,\emptyset)$ is used for the update; 
	its destination is $\dval{dip}:=\sip$.
        We assume this entry is actually inserted in the routing table of \dval{ip}.
	Since $\dval{dip}=\sip=\nhp[\dval{dip}]{\dval{ip}}=\dval{nhip}$,
	the antecedent of the invariant to be proven is not satisfied.
\item[Pro.~\ref{pro:rreq_short}, Line~\ref{rreq:line6}:] 
        The entry $(\oip,\osn,\val,\hops\mathord+1,\sip,*)$ is used for the update; again
        we assume it is inserted into the routing table of node \dval{ip}.
	So $\dval{dip}:=\oip$, $\dval{nhip}:=\sip$,
        $\nsq{\dval{ip}}:=\osn$ and $\dhp{\dval{ip}}:=\hops\mathord+1$.\linebreak[2]
        This information is distilled from a received
	route request message (cf.\ Lines~\ref{aodv:line2} and~\ref{aodv:line8}
	of Pro.~\ref{pro:aodv_short}).
	By Proposition~7.1 of \cite{TR11}, this message was sent before, say in state $N'$;
        by Proposition~7.8 of \cite{TR11}, the sender of this message has identified itself
        correctly, and is $\sip$.

	By \Prop{starcast_ii}, with $\ipc\mathbin{:=}\sip\mathbin=\dval{nhip}$,
	~$\oipc\mathbin{:=}\oip\mathbin=\dval{dip}$, ~$\osnc\mathbin{:=}\osn$~ and
	~$\hopsc\mathbin{:=}\hops$,
        and using that $\ipc = \dval{nhip} \neq \dval{dip} = \oipc$, we get that
	\plat{$\dval{dip}\in \kd[N']{\dval{nhip}}$} and
$$\begin{array}{@{}c@{}}
\fnsqn_{N'}^{\dval{nhip}}(\dval{dip})=\fnsqn_{N'}^{\ipc}(\oipc) > \osnc = \osn\ , \mbox{ or}\\
\fnsqn_{N'}^{\dval{nhip}}(\dval{dip})=\osn \ans
 \fnstatus_{N'}^{\dval{nhip}}(\dval{dip})=\val\ .
\end{array}$$
We first assume that the first line holds.
Then, by \Prop{7.23},
\[
\nsq{\dval{nhip}}
\begin{array}[t]{@{}l@{}}
\mbox{}\geq
\fnnsqn_{N'}^{\dval{nhip}}(\dval{dip})
\geq
\fnsqn_{N'}^{\dval{nhip}}(\dval{dip})\mathord-1 \\
\mbox{} \geq \osn=\nsq{\dval{ip}}\ .
\end{array}
\]
We now assume the second line to be valid. 
From this we conclude
$$\nsq{\dval{nhip}}
\begin{array}[t]{@{}l@{}}
\mbox{}\geq
\fnnsqn_{N'}^{\dval{nhip}}(\dval{dip})
=\fnsqn_{N'}^{\dval{nhip}}(\dval{dip}) \\
\mbox{} =\osn=\nsq{\dval{ip}}\ .
\end{array}$$\vspace{-2mm}
\item[Pro.~\ref{pro:rreq_short}, Lines~\ref{rreq:line18a}, \ref{rreq:line30}:]
In these applications of $\fninv$, the next hop \dval{nhip} is not changed.
Since the invariant has to hold before the execution, it follows that 
$\dval{dip}\in\kd{\dval{nhip}}$ also holds after execution.
Furthermore, in view of Lines~\ref{rreq:line16} and~\ref{rreq:line28},
the route is invalidated while the destination sequence number is incremented.
For this reason the net sequence number stays the same, and the invariant is maintained.
\qed
\end{description}
\end{proof}

\begin{theorem}\label{thm:inv_a}
If, in a state $N$, a node $\dval{ip}\mathop\in\IP$ has a valid entry to
$\dval{dip}$, and the next hop is not \dval{dip} and has a valid $\dval{dip}$-entry as well,
then the latter entry has a larger destination sequence number or an equal one with a
smaller hop count.\\[2mm]
\centerline{$
\begin{array}{@{\hspace{-4mm}}r@{}c@{~}l@{\hspace{-4mm}}}
	&&\dval{dip}\in\akd{\dval{ip}}\cap \akd{\dval{nhip}} \ans\dval{nhip}\not=\dval{dip} \\
	&\Rightarrow&
	\sq{\dval{nhip}}>\sq{\dval{ip}}\vee \big(\sq{\dval{nhip}}\mathbin=\sq{\dval{ip}}\wedge\mbox{}\\
        &&\dhp{\dval{nhip}}\mathbin<\dhp{\dval{ip}}\big)\ ,
\end{array}$
}\\[1mm]
where $\dval{nhip}:=\nhp{\dval{ip}}$, the next hop in the \rte at \dval{ip} for the route to ${\dval{dip}}$.
\end{theorem}
The proof \cite{TR11} is similar to the previous one, but makes use of
Proposition~\ref{prop:inv_nsqn}---including the case where
\plat{$\dval{dip}\notin\akd{nhip}$}---in an essential way.

From Theorem~\ref{thm:inv_a}, we can conclude

\begin{theorem}
AODV is loop free.
\end{theorem}

\begin{proof}
If there were a loop in a routing graph $\RG{N}{\dval{dip}}$, then for
any edge $(\dval{ip},\dval{nhip})$ on that loop one would have
\plat{$\sq{\dval{ip}}\leq\sq{\dval{nhip}}$}, by \Thm{inv_a}.
Hence the value of \plat{$\sq{\dval{ip}}$} is the same for all nodes
\dval{ip} on the loop.
Thus, by \Thm{inv_a}, the sequence numbers keep decreasing when
travelling around the loop, which is impossible.
\end{proof}

\section{Analysing Variants of AODV}
\label{sec:analysis}

In this section, we use {\awn}  to model interpretations and variants of the
AODV protocol. 
Interpretations are just different readings of the RFC, 
variants are the result of modifications to address existing limitations.
Thanks to the use of process algebra, we can easily
adapt the proofs of established correctness properties
of the protocol, such as loop freedom.
This is in contrast to the analysis of variants of existing protocols
via simulations and testbed experiments, where all the work typically has to be
redone from scratch (and even then cannot provide the same level of assurance).

\subsection{Interpretations}\label{ssec:interpretation}
In this section we briefly discuss two ambiguities of the RFC together with possible interpretations. 
More can be found in~\cite{TR11}. 
Of course each interpretation may possibly create routing loops and hence has to be examined separately. 
The specification and the proofs formalised in \awn can easily be adapted.

\subsubsection{Invalidating Routing Table Entries}
We have already presented one contradiction of the AODV RFC in Section~\ref{sec:ambiguities}.
It was based on the question of what would happen if a node that has a
valid \rte for a destination $D$ receives an error message and invalidates the corresponding \rte.
In the same section we also list eight possible resolutions.
The first two, the only interpretations compliant with the RFC,
violate Theorem~\ref{thm:inv_a}, and yield routing loops \cite{ICNP12}.
The same holds for Interpretation \eqref{it:amb_iii}.
As stated before, to guarantee loop freedom one has to create an interpretation of AODV that is
(literally speaking) not compliant with the RFC.
Interpretations~\eqref{it:amb_iv} and \eqref{it:amb_v} are loop free---the proof is identical to the one of \eqref{it:amb_vi}, 
given in Section~\ref{sec:loopfreedom}.
The remaining two resolutions can be proven to be loop free as well~\cite{TR11}.

\subsubsection{Updating with the Unknown Sequence Number}
The AODV RFC~\cite{rfc3561} states that whenever a node receives a forwarded AODV control message from a neighbour (i.e., the neighbour is not the originator of the message), it creates a new or updates an existing {\rte} to that neighbour. In the presented specification, this update is modelled in Lines~\ref{aodv:line10}, \ref{aodv:line14} and \ref{aodv:line18} of Process~\ref{pro:aodv_short}.
In the event a new {\rte} is created, the \dsnf
is set to false to signify that the sequence number corresponding to
the neighbour is unknown. This interpretation is modelled
in~\cite{TR11} and is compliant with the RFC\@. However, in most
implementations of AODV (e.g.\ \cite{AODVUU,CB04}), an un\-known
sequence number is simply represented by the value $0$, rather than by
setting a flag. In the specification of Section~\ref{sec:formalspec}
we follow this approach of using the value~$0$.

Since the RFC does not make the update mechanism clear, different interpretations arise when an existing
valid {\rte} for the neighbour has to be updated. While it is clear that
expiry values for timers associated with the {\rte} will be
updated, it is not clear if a valid sequence number with value $n$ $(\mathord>0)$ will remain unmodified as it is, or be updated to the value
$0$. In order to verify which interpretations are reasonable, we check
which ones satisfy the invariants specified before (e.g., \Thm{inv_a}).

If we assume that an entry $(\dval{dip},0,\val,\dval{hops}',*,*)$
replaces an entry $(\dval{dip},\dval{dsn},\val,\dval{hops},*,*)$ (where $\dval{dsn}\mathord>0$) in a
routing table, it is easy to see that
\Thm{inv_a} is violated.
In fact, it is not hard to exploit this to create a routing loop, since it allows the possibility of decreasing destination sequence numbers~\cite[Sect.\ 9.1]{TR11}. 
This is the interpretation that is implemented by AODV-UIUC and AODV-UCSB.
Our interpretation follows  AODV-UU: an entry $(\dval{dip},\dval{dsn},\val,\dval{hops},\dval{nhip},*)$ is replaced by $(\dval{dip},\dval{dsn},\val,1,\dval{sip},*)$, i.e., in the existing entry only the next hop and the hop count is updated---the sequence number stays the same. 
This interpretation is not in line with the RFC---the RFC never merges information of two routes. However, it is loop free (cf.\ Section~\ref{sec:loopfreedom}).
AODV-ns2 does not perform an update if a \rte already exists, i.e., it uses a version of $\fnupd$ as in Section~\ref{sec:aodv}, but without the fifth clause. This interpretation is also loop free; the proof is identical to the one presented in the previous section.

There are more interpretations possible if the \dsnf
(the flag indicating whether a sequence number is known or unknown) is modelled. 
\begin{enumerate}[(a)]
\item Set the flag to unknown and the sequence number to~$0$. Using the same argument as before, this can yield routing loops.
\item Set the flag to unknown, but keep the destination sequence number stored in the routing table.
This interpretation is indeed loop free and is most likely the intention of the AODV RFC.
\end{enumerate}
Detailed proofs can be found in~\cite{TR11}.

\subsection{Variants}
Let us now turn to variants of AODV\@ and look at (known) shortcomings of the AODV protocol, present possible improvements and then use \awn to verify that the modified AODV is still loop free.
\subsubsection{Non-Optimal Route Selection}\label{ssec:nonoptimal routes}
In AODV's route discovery process, a destination node (or an
intermediate node with an active route to the destination node) will
generate a RREP message in response to a received RREQ message. 
The RREQ message is then discarded and not forwarded. This termination
of the route discovery process at the destination can lead to
other nodes inadvertently creating non-optimal routes to the source
node \cite{MK10}, where route optimality is defined in terms of a
metric, for example hop count. In~\cite{MK10} it is shown
that during the route discovery process in AODV, the only nodes that
generally discover optimal routes to the source and
destination nodes are those lying on the selected route between
the source node and the destination node (or the intermediate node) generating the reply.
All other network nodes receiving the RREQ message (in particular those located
``downstream'' of the destination node) may inadvertently be
creating non-optimal routes to the source node due to the unavoidable
receipt of RREQ messages
over other routes.

 We illustrate this by the example in Figure~\ref{non-optimal-route}.
There, node $S$ wants to find a route to node $D$. It generates and
 broadcasts a RREQ message that is received by its neighbours $D$
 and~$B$. Since~$D$ is the destination, it responds with a
 RREP message. The received RREQ message is discarded and not
 forwarded. On the other hand, $B$ continues to forward its
 received RREQ message, which eventually arrives at $A$. At node
 $A$, a {\rte} is created for the source $S$, with a hop count
 of six. This is clearly not optimal, as $A$ is only two hops away from $S$.
 Due to the discarding of the RREQ message at $D$, $A$ is
 prevented from discovering its optimal route to $S$.

\begin{figure}[t]
\centerline{
\includegraphics[scale=0.94]{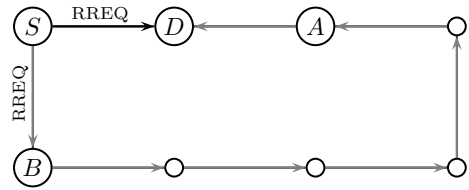}}
\caption{Non-optimal route selection}
\label{non-optimal-route}
\vspace{-4ex}
\end{figure}

A possible modification to AODV to solve this problem is to allow the
destination node to continue to forward the RREQ message. This will
then enable $A$ in Figure~\ref{non-optimal-route} to discover its
optimal route to $S$. In addition, the forwarded RREQ message
from the destination node is modified to include a flag that indicates
a RREP message has already been generated and sent in response to the
former message. This is to prevent other nodes (with active routes to
the destination) from sending a RREP message in response to their
reception of the forwarded RREQ message.

The entire specification of this variant (in \cite{TR11}) differs only in 
five lines from the original---all of which are contained in the process {\RREQ};
the other processes remain unchanged.
The changes introduce the new flag and a case distinction based on that, 
as well as three new broadcasts.
For example, after initiating a route reply at the destination (Process~\ref{pro:rreq_short}, Line \ref{rreq:line14a}),
the route request message is forwarded:%
\newcommand{\mrreq}[8]{\rreqID(#1,#2,#3,#4,#5,#6,#7,#8)}%
\newcommand{\alreadyrep}{\keyw{handled}}%
\newcommand{\true}{\ensuremath{\keyw{true}}}%
\newcommand{\false}{\ensuremath{\keyw{false}}}%
\vspace{-1mm}
\[\mbox{
\small\textbf{broadcast}(\mrreq{\hops$\mathord+1$}{\rreqid}{\dip}{\dsn}{\oip}{\osn}{\ip}{\true})\ ,
}\vspace{-1mm}\]
where the last component of the RREQ message is the newly introduced flag.
The proofs of important properties (e.g., loop freedom) are still valid. 
The proofs of the invariants proceed by examining lines in our
processes where the invariant might be invalidated: The proof of
Proposition~\ref{prop:starcast_ii}, which can be found in \cite{TR11}, checks all occurrences
of sending a RREQ message, and the proofs of Proposition~\ref{prop:inv_nsqn} and \Thm{inv_a} check all
occurrences of {\fnupd} and {\fninv}. For the former, three new
{\bf broadcast}-commands have to be examined;
however these cases are similar to the {\bf broadcast} 
already implemented in the original process {\RREQ} (Line~\ref{rreq:line34}).
For the latter, no extra effort is needed, as the modification does
not involve occurrences of {\fnupd} and {\fninv}.

\begin{figure*}[t]
\renewcommand{\msg}[3]{#1${}_{\!\!\;\MakeUppercase{#2}\!\!\:\shortrightarrow\!\MakeUppercase{#3}}$}
\subfloat[][\rule{0pt}{10pt}$S$~broadcasts a RREQ~to~$D$]{\raisebox{5.25mm}{\hspace{4.25mm}\includegraphics[scale=1]{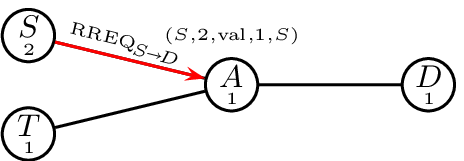}\hspace{4.25mm}}}
\hfill 
\subfloat[][\rule{0pt}{10pt}$T$~broadcasts a RREQ~to~$D$.]{\raisebox{5.25mm}{\hspace{4.25mm}\includegraphics[scale=1]{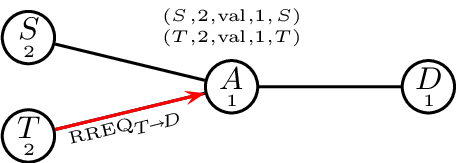}\hspace{4.25mm}}}
\hfill
\subfloat[][\rule{0pt}{10pt}$A$ re-broadcasts \msg{RREQ}{s}{d}]{\includegraphics[scale=1]{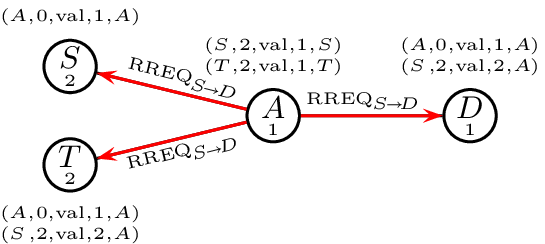}}

\vspace{2mm}

\subfloat[][\rule{0pt}{10pt}$A$ re-broadcasts  \msg{RREQ}{t}{d}]{\includegraphics[scale=1]{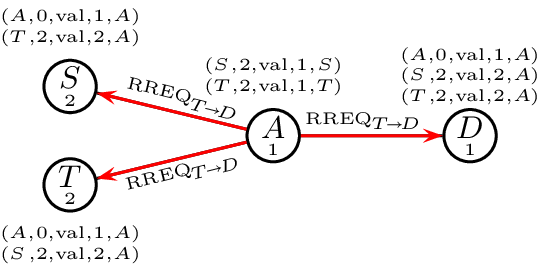}}
\hfill
\subfloat[][\rule{0pt}{10pt}$D$ unicasts \msg{RREP}{s}{d} to $A$;\\
\phantom{(e) }$A$ forwards to $S$]{\includegraphics[scale=1]{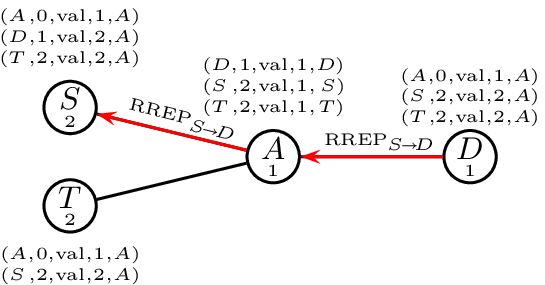}}
\hfill
\subfloat[][\rule{0pt}{10pt}$D$ unicasts \msg{RREP}{t}{d} to $A$;\\
\phantom{(f) }$A$ drops \msg{RREP}{t}{d}]{\includegraphics[scale=1]{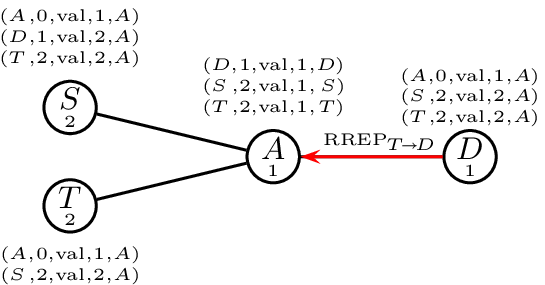}}
\caption{Failure of route discovery process}
\label{fig:topo}
\vspace{-1ex}
\end{figure*}

\subsubsection{Failure of Route Discovery Process}\label{ssec:failure}
In AODV's route discovery process, a RREP message from the destination
node is unicast back along the selected route towards the
originator of the RREQ message. Every intermediate node on the
selected route will process the RREP message and, in most cases,
forwards it towards the originator node. However, there is a
possibility that the RREP message is discarded at an intermediate
node and hence results in the originator node not receiving a
reply. The discarding of the RREP message is due to the RFC
specification of AODV \cite{rfc3561} stating that an intermediate node
only forwards the RREP message if it is not the originator node
\emph{and} it has created or updated a {\rte} to the destination
node described in the RREP message. The latter requirement means that
if a valid {\rte} to the destination node already exists, and is
not updated when processing the RREP message, then the intermediate
node will not forward the message. We illustrate this problem
with an example.\footnote{\raggedright A slightly different example was
given on the MANET mailing list
\url{http://www.ietf.org/mail-archive/web/manet/current/msg05702.html}}

Figure~\ref{fig:topo} shows a four-node topology. In Figures~\ref{fig:topo}(a) and \ref{fig:topo}(b), source nodes $S$ and $T$, respectively, initiate a route discovery process to search for a route to~$D$. When generating a RREQ message, the source node increments its sequence number before inserting it into the message. 
In processing the RREQ messages from $S$ and $T$, node $A$ creates \rtes in its routing table.

In Figures~\ref{fig:topo}(c) and \ref{fig:topo}(d), node $A$ re-broadcasts the RREQ messages that it received previously. The destination node $D$ receives the RREQ messages, and creates corresponding entries in its routing table. In Figure~\ref{fig:topo}(e), $D$ replies with a RREP message in response to the RREQ message from node $S$. Since the RREQ message from~$S$ does not contain any information on the destination sequence number for $D$, node $D$ inserts its sequence number of $1$ into the RREP message. This message is then processed by $A$ (a \rte to $D$ is created) and forwarded to $S$. 

Similarly, in Figure~\ref{fig:topo}(f), node $D$ replies with a RREP message in response to the RREQ message from node $T$. 
Since the RREQ message from node $T$ does not contain any information on the destination sequence number for node $D$, node $D$ again inserts its sequence number of $1$ into its RREP message. When the intermediate node $A$ receives the RREP message, it processes the message. However, the existing {\rte} that node $A$ already has for destination node $D$ contains the same information (same destination sequence number and same hop count) as in the received RREP message. Therefore, node $A$ does not update its {\rte} for node $D$ and thus, according to the RFC specification, will not forward the RREP message to the source node $T$. This then leads to an unsuccessful route discovery process for node $T$.

A solution to this problem is to require intermediate nodes to forward 
\emph{all} RREP messages that they receive using the newest available information on 
the route to the destination node: if the node's  routing table contains an entry for the destination node
that is valid and  ``fresher'' than that in the RREP message, the intermediate node
updates the contents of the RREP message to reflect this.
The intermediate node $A$ will then forward a RREP message containing up-to-date information on the destination node~$D$.

As in Section~\ref{ssec:nonoptimal routes},  this solution 
does not 
violate any of the
invariants; again this follows by adapting the original proofs~\cite{TR11}.

\section{Related Work}
\label{sec:related}
Previous attempts to prove loop-freedom of AODV have been reported
in \cite{AODV99,BOG02,ZYZW09}, but none of these proofs are complete and valid for the current version of AODV \cite{rfc3561}.
\cite{AODV99} fails to consider the effect of RERR messages, which can result in routing loops;
\cite{BOG02} analyses an earlier draft of AODV and uses an invariant that does not hold when following the AODV RFC;
and \cite{ZYZW09} only considers a restricted version of AODV, not
covering the important case of route replies by intermediate nodes. 
Details of the limitations of these proof attempts are provided in~\cite{ICNP12}.

Graph Transformation Systems were used in~\cite{SWJ08} to model DYMO
v10, a protocol derived from AODV\@. The paper provides a
semi-algorithm, based on graph rewriting, which was used to verify
loop-freedom for DYMO\@. Model checking is also used 
to verify properties of routing protocols for WMNs. For example, \cite{WPP04} shows loop freedom of the ad-hoc protocol LUNAR, for fixed topologies
or set changes in the topology. Model checking in general lacks the ability to
verify protocols for an arbitrary and changing topology.
It is used to check specific scenarios only.

\section{Conclusions}\label{sec:conclude}
In this paper, we have presented a complete and accurate model of 
the core functionality of AODV using the process algebra {\awn}, which has been tailored 
specifically for the formal modelling of 
wireless  mesh networks and MANETs. The unique set of features 
and primitives of {\awn} allows the creation of accurate, concise and 
readable models of relatively complex and practically relevant 
network protocols, which we have demonstrated with AODV\@.
This is in contrast to some prior related work, which either 
modelled only very simple protocols, or modelled only a 
subset of the functionality of relevant WMN or MANET routing 
protocols.

The currently predominant practice of informally 
specifying WMN and MANET protocols via English prose has a potential 
for ambiguity and inconsistent interpretation. The ability to provide 
a formal and unambiguous specification of such protocols via {\awn} 
is a significant benefit in its own right.
We have demonstrated how  {\awn} 
can be used as a basis for reasoning about
critical protocol correctness properties, illustrated with the example of loop freedom. 
We have further shown how relevant proofs can relatively easily be 
adapted to protocol variants. 
In contrast to protocol evaluation using simulation, test-bed experiments or model checking, 
where only a finite number of specific network scenarios can be considered, our reasoning with {\awn} is generic and the the proofs hold for any possible network scenario in terms of topology and traffic pattern. None of the experimental protocol evaluation approaches can deliver this high degree of assurance about protocol behaviour.

\bibliographystyle{abbrv}

\balancecolumns
\end{document}